\documentclass[noderivs,copyright]{eptcs}
\providecommand{\event}{AFL 2017} % Name of the event you are submitting to

\usepackage{underscore}           % Only needed if you use pdflatex.

\usepackage{mathrsfs}
\usepackage{amssymb}
\usepackage{amsmath}
\usepackage{graphicx}
\usepackage{complexity}

\newcommand{\nat}{\mathbb{N}}
\newcommand{\ints}{\mathbb{Z}}

\newcommand{\sset}[2]{\left\{\,#1\mathrel{\left\vert\vphantom{#1#2}\right.} #2\,\right\}}

\newcommand*{\qed}{\raisebox{0.5ex}[0ex][0ex]{\framebox[1ex][l]{}}}
\newtheorem{theorem}{Theorem}
\newtheorem{lemma}[theorem]{Lemma}

\newenvironment{proof}{%
  \par\noindent
  {\rmfamily\itshape\mdseries Proof\/}:\hspace{\labelsep}\ignorespaces}%
  {\mbox{}\nolinebreak\hfill~%
  {\qed}
  \medbreak
}

\title{On the Descriptional Complexity of Operations\\ on Semilinear Sets}
\author{
Simon Beier \qquad\qquad Markus Holzer \qquad\qquad Martin Kutrib
\institute{Institut f\"ur Informatik, Universit\"at Giessen,\\
Arndtstr.~2, 35392 Giessen, Germany}
\email{\quad \texttt{$\{$simon.beier,holzer,kutrib$\}$@informatik.uni-giessen.de}}
}
\def\titlerunning{On the Descriptional Complexity of Operations on Semilinear Sets}
\def\authorrunning{S. Beier, M. Holzer \& M. Kutrib}
\begin{document}
\maketitle

\begin{abstract}
We investigate the descriptional complexity of operations on
  semilinear sets. Roughly speaking, a semilinear set is the finite
  union of linear sets, which are built by constant and period
  vectors. The interesting parameters of a semilinear set are: (i) the
  maximal value that appears in the vectors of periods and constants and (ii) the
  number of such sets of periods and constants necessary to describe the semilinear set under
  consideration.
  More precisely, we prove upper bounds on the union, intersection,
  complementation, and inverse homomorphism. In particular, our result
  on the complementation upper bound answers an open problem from
  [\textsc{G.~J.~Lavado}, \textsc{G.~Pighizzini}, \textsc{S.~Seki}:
  Operational State Complexity of Parikh Equivalence, 2014].
\end{abstract}

\section{Introduction}
\label{sec:intro}

A subset of~$\nat^k$, where $\nat$ refers to the set of non-negative
integers, of the form
$$
L(C,P)=\sset{\vec{c}+\sum_{\vec{x}_i\in P} \lambda_i\cdot
  \vec{x}_i}{\mbox{$\vec{c}\in C$ and $\lambda_i\in\nat$}},
$$
for finite sets of periods and constants $P,\, C\subseteq\nat^k$, is
said to be \emph{linear} if $C$ is a singleton set. In this case we
just write $L(\vec{c},P)$, where $\vec{c}$ is the constant
vector. This can be seen as a straightforward generalization of an
arithmetic progression allowing multiple differences. Moreover, a
subset of~$\nat^k$ is said to be \emph{semilinear} if it is a finite
union of linear sets. Semilinear sets were extensively studied in the
literature and have many applications in formal language and automata
theory.

Let us recall two famous results from the very beginning of computer
science, where semilinear sets play an important role.  The Parikh
image of a word $w\in\Sigma^*$ is defined as the function
$\psi:\Sigma^*\rightarrow\nat^{|\Sigma|}$ that maps~$w$ to a vector
whose components are the numbers of occurrences of letters
from~$\Sigma$ in~$w$. Parikh's theorem states that the Parikh image
of every context-free language $L$, that is, $\{\, \psi(w)\mid w\in
L\,\}$, is a semilinear set~\cite{Pa66}.  A direct application of
Parikh's theorem is that every context-free language is letter
equivalent to a regular language. Another famous result on semilinear
sets is their definability in Presburger arithmetic~\cite{GiSp66},
that is, the first order theory of natural numbers with addition but
without multiplication. Since Presburger arithmetic is decidable,
corresponding questions on semilinear sets are decidable as well,
because the conversion between semilinear set representations by
vectors and Presburger formulas and \textit{vice versa} is effectively
computable.
 
Recently, semilinear sets appeared particularly in two different research
directions from automata theory. The first research direction is that
of jumping automata, a machine model for discontinuous information processing,
recently introduced in~\cite{MeZe12}. Roughly speaking, a
jumping finite automaton is an ordinary finite automaton, which is
allowed to read letters from anywhere in the input string, not
necessarily only from the left of the remaining input.  Since a
jumping finite automaton reads the input in a discontinuous fashion,
obviously, the order of the input letters does not matter. Thus, only
the number of symbols in the input is important. In this way, the behavior
of jumping automata is somehow related to the notions of Parikh image
and Parikh equivalence. As already mentioned regular and context-free
languages cannot be distinguished \textit{via} Parikh equivalence,
since both language families have semilinear Parikh images. This is in
fact the starting point of the other recent research direction, the
investigation of several classical results on automata conversions and
operations subject to the notion of Parikh equivalence. For instance,
in~\cite{LPS13} it was shown that the cost of the conversion of an
$n$-state nondeterministic finite automaton into a Parikh equivalent
deterministic finite state device is of order~$e^{\Theta(\sqrt{n\ln
    n})}$---this is in sharp contrast to the classical result on
finite automata determinization which requires~$2^n$ states in the
worst case. A close inspection of these results reveals that there is a
nice relation between Parikh images and Parikh equivalence of regular
languages and jumping finite automata \textit{via} semilinear sets.
Thus one can read the above mentioned results as results on semilinear
sets as well.

Here we investigate the descriptional complexity of the operation
problem on semilinear sets.  Recall that semilinear sets are closed
under Boolean operations.
The operands of the operations are semilinear sets of the form  $\bigcup_{i\in I}L(\vec{c}_i,P_i)\subseteq \nat^k$. Our resulting semilinear sets are of the form $S=\bigcup_{j\in J}L(C_j,Q_j)\subseteq \nat^k$.
We investigate upper bounds for the cardinality $|J|$ of the index set and for the norms $||Q_j||$ and $||C_j||$, these are the maximal values that appear in the
vectors of periods $Q_j$ and constants $C_j$. From this, one can automatically get upper bounds for the cardinalities of periods and constants through $|Q_j|\le(||Q_j||+1)^k$ and $|C_j|\le(||C_j||+1)^k$. One can also write the resulting set $S$ in the form $S=\bigcup_{j\in J,\\ \vec{c}\in C_j}L(\vec{c},Q_j)$, which is a finite union of linear sets. In this form the index set has cardinality $\sum_{j\in J}|C_j|$.
Upper bounds are proved for the
Boolean operations and inverse homomorphism on semilinear sets. For
instance, roughly speaking we show that intersection increases the size
description polynomially, while complementation increases it double
exponentially. A summary of our results can be found in
Table~\ref{tab:results}.
The precise bound of the former result improves a recent result
shown in~\cite{LPS14}, and the latter result on the complementation
answers an open question stated in~\cite{LPS14}, too. 

\begin{table}[tbh]
  \begin{center}
%\centering
  \begin{tabular}{|l||c|c|}\hline\hline
     & \multicolumn{2}{|c|}{Parameters of the resulting semilinear set $\bigcup_{j\in J}L(C_j,Q_j)$}\\\cline{2-3}
    Operation & $|J|$ & $\max\{||C_j||,||Q_j||\}$ \\\hline\hline
    Union & $|I_1|+|I_2|$ & $\nu$ \\\hline
    Intersection & $|I_1|\cdot|I_2|$ & $O(m^2\nu^{k+2}+\nu)$ \\\hline
    Complementation & $2^{(\nu+2)^{O(m)}\cdot|I_1|^{\log(3k+2)}}$ & $2^{(\nu+2)^{O(m)}\cdot|I_1|^{\log(3k+2)}}$ \\\hline
		Inverse Homom. & $|I_1|$ & $O\left((||H||+1)^{\min(k_1,k)}(m+1)(\nu+1)^{k+1}\right)$ \\ \hline\hline
  \end{tabular}
  \end{center}
  \caption{Descriptional complexity results on the operation problem for semilinear subsets of~$\nat^k$. We assume~$k$ to be a constant in this table.
	The operands of the operations are semilinear sets of the form~$\bigcup_{i\in I_\epsilon}L(\vec{c}_i,P_i)\subseteq\nat^k$, where $\epsilon\in\{1,2\}$ for the first two operations and $\epsilon=1$ for the last two operations. The parameter $\nu$ is the maximal value that appears in the
vectors of periods and constants in the operands. The parameter~$m$ is the maximal cardinality $|P_i|$ of all the period sets appearing in the operands.
The inverse homomorphism is given by the matrix $H\in\nat^{k\times k_1}$, where~$k_1$ is also assumed to be a constant in this table. The parameter $||H||$ is the maximal value that appears in $H$.
	}
  \label{tab:results}
\end{table}

It is worth
mentioning that independently in~\cite{ChHa16} the operation problem
for semilinear sets over the integers~$\mathbb{Z}$ were studied. The
obtained results there rely on involved decomposition techniques for
semilinear sets. In contrast to that, our results are obtained by
careful inspections of the original proofs on the closure
properties.
An application of the presented results on semilinear sets to the
descriptional complexity of jumping automata and finite automata
subject to Parikh equivalence is given in~\cite{BHK17a}.

\section{Preliminaries}
\label{sec:defs}

Let~$\ints$ be the set of integers and
$\nat=\{0,1,2,\ldots\}$ be the set of non-negative integers.
For the notion of semilinear sets we follow the notation of Ginsburg
and Spanier~\cite{GiSp64}. For a natural number $k\geq 1$ and
finite~$C,P\subseteq\nat^k$ let~$L(C,P)$ denote the
subset
$$
L(C,P)=\sset{\vec{c}+\sum_{\vec{x}_i\in P} \lambda_i\cdot
  \vec{x}_i}{\mbox{$\vec{c}\in C$ and $\lambda_i\in\nat$}}
$$
of~$\nat^k$. Here the~$\vec{c}\in C$ are called the \emph{constants}
and the~$\vec{x}_i\in P$ the \emph{periods}. If~$C$ is a singleton set we call~$L(C,P)$ a \emph{linear} subset of
$\nat^k$. In this case we simply write~$L(\vec{c},P)$ instead of
$L(\{\vec{c}\},P)$.  A subset of~$\nat^k$ is said to be
\emph{semilinear} if it is a finite union of linear subsets.  We
further use~$|P|$ to denote the size of a finite subset
$P\subseteq\nat^k$ and $||P||$ to refer to the value
$\max\{\,||\vec{x}||\mid \vec{x}\in P\,\}$, where~$||\vec{x}||$ is the
maximum norm of~$\vec{x}$, that is,
\mbox{$||(x_1,x_2,\ldots,x_k)||=\max\{\,|x_i|\mid 1\leq i\leq k\,\}$.}
Observe, that $$|P|\leq (||P||+1)^k.$$ Analogously we write $||A||$
for the maximum norm of a matrix $A$ with entries in~$\mathbb{Z}$,
i.e. the maximum of the absolute values of all entries of $A$.  The
elements of~$\nat^k$ can be partially ordered by the $\leq$-relation
on vectors. For vectors $\vec{x},\vec{y}\in\nat^k$ we write
$\vec{x}\leq \vec{y}$ if all components of~$\vec{x}$ are less or equal
to the corresponding components of~$\vec{y}$. In this way we
especially can speak of \emph{minimal elements} of subsets
of~$\nat^k$. In fact, due to~\cite{Di13} every subset of~$\nat^k$ has
only a finite number of minimal elements.

Most results on the descriptional complexity of operations on
semilinear sets is based on a size estimate of minimal solutions of
matrix equations. We use a result due to~\cite[Theorem 2.6]{Hu80},
which is based on~\cite{GaSi78}, and can slightly be improved by a
careful inspection of the original proof. The generalized result reads
as follows:

\begin{theorem}\label{thm:size-of-minimal-solutions}
  Let $s,t\geq 1$ be integers, $A\in\ints^{s\times t}$ be a matrix of
  rank~$r$, and $\vec{b}\in\ints^s$ be a vector. Moreover, let~$M$
  be the maximum of the absolute values of the $r\times r$
  sub-determinants of the extended matrix~${(A\mid \vec{b})}$, and
  $S\subseteq\nat^t$ be the set of minimal elements of
  ${\{\,\vec{x}\in\nat^t\setminus\{\vec{0}\}\mid
    A\vec{x}=\vec{b}\,\}}$. Then $||S||\leq (t+1)\cdot M$.
\end{theorem}

We will estimate the value of the above mentioned (sub)determinants
with a corollary of Hadamard's inequality:

\begin{theorem}\label{thm:hadamard}
  Let $r\geq 1$ be an integer, $A\in\ints^{r\times r}$ be a matrix,
  and~$m_i$, for $1\leq i\leq r$, be the maximum of the absolute
  values of the entries of the $i$th column of~$A$.  Then ${|\det(A)|\le
    r^{r/2}\prod_{i=1}^r m_i}$.
\end{theorem}
 
\section{Operational complexity of semilinear sets}

In this section we consider the descriptional complexity of operations
on semilinear sets. We investigate the Boolean
operations union, intersection, and complementation
w.r.t.~$\nat^k$. Moreover, we also study the operation of inverse
homomorphism on semilinear sets.  

\subsection{Union on semilinear sets}

For the union of semilinear sets, the following result is straightforward.

\begin{theorem}\label{thm:union-semilinear-set}
  Let $\bigcup_{i\in I}L(\vec{c}_i,P_i)$ and $\bigcup_{j\in
    J}L(\vec{c}_j,P_j)$ be semilinear subsets of $\nat^k$, for some
  $k\geq 1$. Assume that~$I$ and~$J$
  are disjoint finite index sets.  Then the union
$$\left(\bigcup_{i\in I}L(\vec{c}_i,P_i)\right)\cup\left(\bigcup_{j\in J}L(\vec{c}_j,P_j)\right)=
\bigcup_{i\in I\cup J}L(\vec{c}_i,P_i)$$ can be described by a
semilinear set with index sets size $|I|+|J|$, the maximal number of
elements~$m=\max_{i\in I\cup J}|P_i|$ in the period sets, and the
entries in the constant vectors are at most $\ell=\max_{i\in I\cup
  J}||\vec{c}_i||$ and in the period vectors at most $n=\max_{i\in
  I\cup J}||P_i||$.\hfill\qed
\end{theorem}

Thus, the size increase for union on semilinear sets is only linear
with respect to all parameters.

\subsection{Intersection of semilinear sets}

Next we consider the intersection operation on semilinear sets. The
outline of the construction is as follows: we analyse the proof that
semilinear sets are closed under intersection
from~\cite[Theorem~6.1]{GiSp64}. Due to distributivity it suffices to
look at the intersection of linear sets. Those coefficients of the
periods of our linear sets, which deliver a vector in the
intersection, are described by systems of linear equations.  For the
intersection of the linear sets we get a semilinear set, where the
periods and constants are built out of the minimal solutions of these
systems of equations.  We will estimate the size of the minimal
solutions with the help of
Theorems~\ref{thm:size-of-minimal-solutions} and~\ref{thm:hadamard} in
order to obtain upper bounds for the norms of the resulting periods
and constants.
\begin{sloppypar}
\begin{theorem}\label{thm:semilinear-set-intersection} 
  Let $\bigcup_{i\in I}L(\vec{c}_i,P_i)$ and $\bigcup_{j\in
    J}L(\vec{c}_j,P_j)$ be semilinear subsets of $\nat^k$, for some
  $k\geq 1$. Assume that~$I$ and~$J$
  are disjoint finite index sets.  We set $n=\max_{i\in I\cup
    J}||P_i||$,\ $m=\max_{i\in I\cup J}|P_i|$,\ and $\ell=\max_{i\in
    I\cup J}||\vec{c}_i||$. Then for every $(i,j)\in I\times J$ there
  exist $P_{i,j},\,C_{i,j}\subseteq\nat^k$ with
\begin{align*}
||P_{i,j}||&\le3m^2k^{k/2}n^{k+1},\\
||C_{i,j}||&\le(3m^2k^{k/2}n^{k+1}+1)\ell,
\end{align*}
and $\bigl(\bigcup_{i\in I}L(\vec{c}_i,P_i)\bigr)\cap\left(\bigcup_{j\in J}L(\vec{c}_j,P_j)\right)=\bigcup_{(i,j)\in I\times J}L(C_{i,j},\,P_{i,j})$.
\end{theorem}
\end{sloppypar}

\begin{proof}
  We analyse the proof that semilinear sets are closed under
  intersection from~\cite[Theorem~6.1]{GiSp64}. Let $i\in I$ and $j\in
  J$ be fixed and let
  \mbox{$P_i=\{\vec{x}_1,\vec{x}_2,\dots,\vec{x}_p\}$,} and
  $P_j=\{\vec{y}_1,\vec{y}_2,\dots,\vec{y}_q\}$. Denote by~$X$ and~$Y$
  the subsets of~$\nat^{p+q}$ defined by
\begin{align*}
X &=\sset{(\lambda_1,\dots,\lambda_p,\mu_1,\dots,\mu_q)\in\nat^{p+q}}{\vec{c}_i+\sum_{r=1}^p\lambda_r\vec{x}_r=\vec{c}_j+\sum_{s=1}^q\mu_s\vec{y}_s}\\
\noalign{\hbox{and}}
 Y &=\sset{(\lambda_1,\dots,\lambda_p,\mu_1,\dots,\mu_q)\in\nat^{p+q}}{\sum_{r=1}^p\lambda_r\vec{x}_r=\sum_{s=1}^q\mu_s\vec{y}_s}.
\end{align*}
Let~$C$ and~$P$ be the sets of minimal elements of~$X$ and~$Y\setminus\{\vec{0}\}$. In the proof of~\cite[Theorem~6.1]{GiSp64} it was shown that $X=L(C,P)$.

In order to estimate the size of~$||C||$ and~$||P||$ we use an
alternative description of the vectors in~$X$ and~$Y$ in terms of
matrix calculus. Let us define the matrix $H=(\vec{x}_1\mid
\vec{x}_2\mid\cdots\vec{x}_p\mid -\vec{y}_1\mid -\vec{y}_2\mid
\cdots\mid -\vec{y}_q)$ in~$\ints^{k\times(p+q)}$. Then it is easy to
see that
\begin{align*}
\vec{x}\in X\quad & \mbox{if and only if}\quad H\vec{x}=\vec{c}_j-\vec{c}_i,\\
\noalign{\hbox{and}}
\vec{y}\in Y\quad & \mbox{if and only if}\quad H\vec{y}=\vec{0}.
\end{align*}
With ${||P_i||,\,||P_j||\le n}$, we derive from
Theorem~\ref{thm:hadamard} that the maximum of the absolute values of
any $r\times r$ sub-determinant, for $1\leq r\leq k$, of the extended
matrix $(H\mid \vec{0})$ is bounded from above by $k^{k/2}n^k$,
because the maximum of the absolute values of the entries of the whole
extended matrix $(H\mid \vec{0})$ is~$n$. Then by
Theorem~\ref{thm:size-of-minimal-solutions} we conclude that
$$
||P||\le(p+q+1)k^{k/2}n^k\le3mk^{k/2}n^k.
$$
Analogously we can estimate the
value of the maximum of the absolute values of any $r\times r$
sub-deter\-minant, for $1\leq r\leq k$, of the extended matrix $(H\mid
\vec{c}_j-\vec{c}_i)$ by Theorem~\ref{thm:hadamard}. It is bounded by
$k^{k/2}n^k\ell$, because the maxima of the
absolute values of the columns of $(H\mid \vec{c}_j-\vec{c}_i)$ are
bounded by~$n$ and $\ell$. Thus we
have 
$$||C||\le(p+q+1)k^{k/2}n^k\ell\le3mk^{k/2}n^k\ell$$
by Theorem~\ref{thm:size-of-minimal-solutions}.

Let~$\tau:\nat^{p+q}\to\nat^k$ be the linear function given by
$(\lambda_1,\dots,\lambda_p,\mu_1,\dots,\mu_q)\mapsto\sum_{r=1}^p\lambda_r\vec{x}_r$.
Then we have~${L(\vec{c}_i,P_i)\cap
  L(\vec{c}_j,P_j)=\vec{c}_i+\tau(X)}$. The linearity of~$\tau$
implies that~$\tau(X)$ is equal to the semilinear set~${L(\tau(C),\tau(P))}$
(see, for example,~\cite{GiSp64}). So we get $L(\vec{c}_i,P_i)\cap
L(\vec{c}_j,P_j)=L(\vec{c}_i+\tau(C),\tau(P))$. Because of~$p\le m$
and~${||P_i||\le n}$ we obtain
\begin{align*}
||\tau(P)|| & \le m\cdot||P||\cdot n \le 3m^2k^{k/2}n^{k+1}\\
\noalign{\hbox{and}}
||\tau(C)|| & \le m\cdot||C||\cdot n \le 3m^2k^{k/2}n^{k+1}\ell.
\end{align*}
It follows that $||\vec{c}_i+\tau(C)||\le \ell+||\tau(C)||=(3m^2k^{k/2}n^{k+1}+1)\ell$.

Because $\bigl(\bigcup_{i\in
  I}L(\vec{c}_i,P_i)\bigr)\cap\left(\bigcup_{j\in
    J}L(\vec{c}_j,P_j)\right)$ is equal to the semilinear set
$\bigcup_{(i,j)\in I\times J}L(\vec{c}_i,P_i)\cap L(\vec{c}_j,P_j)$
our theorem is proved.
\end{proof}

The index set of the semilinear set for the intersection has size
$|I|\cdot|J|$ and the norms of the periods and constants are in
$O(m^2\nu^{k+2}+\nu)$ if dimension $k$ is constant. Here $\nu$ is the maximum
of $n$ and $\ell$, which means that it is the maximum norm appearing in the two
operands of the intersection. So the size increase for intersection is
polynomial with respect to all parameters. 

Now we turn to the intersection of more than two semilinear sets.  The
result is later utilized to explore the descriptional complexity of
the complementation.  First we have to deal with the intersection of
two semilinear sets of the form $\bigcup_{i\in I}L(C_i,P_i)$ instead
of $\bigcup_{i\in I}L(\vec{c}_i,P_i)$ as in the previous theorem. The
following lemma is proved by writing a semilinear set of the form
$L(C_i,P_i)$ as $\bigcup_{\vec{c}_i\in C_i}L(\vec{c}_i,P_i)$ and
applying Theorem~\ref{thm:semilinear-set-intersection}.

\begin{sloppypar}
\begin{lemma}\label{lem:semilinear-set-intersection}
  Let $\bigcup_{i\in I}L(C_i,P_i)$ and $\bigcup_{j\in J}L(C_j,P_j)$ be
  semilinear subsets of~$\nat^k$, for some $k\geq 1$. Assume that~$I$ and~$J$ are disjoint finite index
  sets. We set $p=\max\{|I|,\,|J|\}$,\ $n=\max_{i\in I\cup J}||P_i||$,
  and $\ell=\max_{i\in I\cup J}||C_i||$.  
Define 
$a_k=4^{k+1}k^{k/2}$.
Then there exists an index
  set~$H$ with $$|H|\le p^2(\ell+1)^{2k}$$ such that, for each $h\in H$,
  there are $P_h,\,C_h\subseteq\nat^k$ with
\begin{align*}
||P_h||&\le a_kn^{3k+1},\\
||C_h||&\le (a_kn^{3k+1}+1)\ell,
\end{align*}
and $\Bigl(\bigcup_{i\in I}L(C_i,P_i)\Bigr)\cap\left(\bigcup_{j\in
  J}L(C_j,P_j)\right)=\bigcup_{h\in H}L(C_h,\,P_h)$.
\end{lemma}
\end{sloppypar}

\begin{proof}
  Let $i\in I,\,j\in J,\,\vec{c}\in C_i,\textup{ and }\vec{d}\in C_j$
  be fixed. The proof of
  Theorem~\ref{thm:semilinear-set-intersection} shows that there exist~$C_{i,\,j,\,\vec{c},\,\vec{d}},\,P_{i,\,j,\,\vec{c},\,\vec{d}}\subseteq\nat^k$
  with
\begin{align*}
||P_{i,\,j,\,\vec{c},\,\vec{d}}||&\le3m^2k^{k/2}n^{k+1},\\
||C_{i,\,j,\,\vec{c},\,\vec{d}}||&\le(3m^2k^{k/2}n^{k+1}+1)\ell,
\end{align*}
and $L(\vec{c},P_i)\cap
L(\vec{d},P_j)=L(C_{i,\,j,\,\vec{c},\,\vec{d}},\,P_{i,\,j,\,\vec{c},\,\vec{d}})$,
where $m$ is the maximum of $|P_i|$ and $|P_j|$. Since~${P_i,\,P_j\subseteq\nat^k}$, we have $m\le(n+1)^k\le(2n)^k$, for
$n>0$. This gives us
$$||P_{i,\,j,\,\vec{c},\,\vec{d}}||\le3m^2k^{k/2}n^{k+1}\le3(2n)^{2k}k^{k/2}n^{k+1}=3\cdot4^kk^{k/2}n^{3k+1}\le a_kn^{3k+1}$$
and $||C_{i,\,j,\,\vec{c},\,\vec{d}}||\le(a_kn^{3k+1}+1)\ell$. With
$$\left(\bigcup_{i\in I}L(C_i,P_i)\right)\cap\left(\bigcup_{j\in J}L(C_j,P_j)\right)=\bigcup_{(i,j)\in I\times J}L(C_i,P_i)\cap L(C_j,P_j)$$
and
$$L(C_i,P_i)\cap L(C_j,P_j)=\left(\bigcup_{\vec{c}\in C_i}L(\vec{c},P_i)\right)\cap\left(\bigcup_{\vec{d}\in C_j}L(\vec{d},P_j)\right)
=\bigcup_{(\vec{c},\vec{d})\in C_i\times C_j}L(\vec{c},P_i)\cap L(\vec{d},P_j)$$
our result is proven because of $|C_i\times
C_j|\le(\ell+1)^{2k}$.
\end{proof}

Now we present the result on the intersection of a finite number of semilinear sets.

\begin{theorem}\label{theo:semilinear-set-intersection-many}
  Let $k\geq 1$ and $X\ne\emptyset$ be a finite index set. For every
  $x\in X$ let $\bigcup_{i\in I_x}L(C_i,P_i)$ be a semilinear subset
  of $\nat^k$. Assume that $I_x$,
  $I_y$ are disjoint finite index sets for~$x,\,y\in X$ with $x\ne y$.
  We set~$n=\max_{x\in X,\,i\in I_x}||P_i||$ and $\ell=\max_{x\in
    X,\,i\in I_x}||C_i||$. Define $p=\max_{x\in X}|I_x|$,
  $q=\lceil\log_2|X|\rceil$, and $a_k=4^{k+1}k^{k/2}$.  Then there exists an index set~$J$ with
	\begin{equation}\label{eqn:est-J}
	|J|\le p^{2^q}(\ell+1)^{k\cdot 2^{q+1}}(a_kn+1)^{4(3k+2)^{q+1}},
	\end{equation}
such that, for each $j\in J$, there are
	$P_j,\,C_j\subseteq\nat^k$ with
\begin{align*}
||P_j||&\le(a_kn)^{(3k+1)^q},\\
||C_j||&\le(a_kn+1)^{(3k+2)^q}\ell,
\end{align*}
and $\bigcap_{x\in X}\left(\bigcup_{i\in I_x}L(C_i,P_i)\right)=\bigcup_{j\in J}L(C_j,P_j)$.
\end{theorem}

\begin{proof}
We prove this by induction on $q$. For $q=0$ we have $|X|=1$, so let $X$ be the set $\{x\}$. Then we choose $J=I_x$ and get
\begin{align*}
|J|&=p\le p^1(\ell+1)^{2k}(a_kn+1)^{4(3k+2)^1}=p^{2^q}(\ell+1)^{k\cdot 2^{q+1}}(a_kn+1)^{4(3k+2)^{q+1}}\\
\noalign{\hbox{and}}
||P_j||&\le n\le(a_kn)^1=(a_kn)^{(3k+1)^q},\\
||C_j||&\le \ell\le(a_kn+1)^1\ell=(a_kn+1)^{(3k+2)^q}\ell
\end{align*}
for every $j\in J=I_x$. This proves the statement for $q=0$.

For $q=1$ we have $|X|=2$. In this case our statement follows directly from Lemma~\ref{lem:semilinear-set-intersection}.
Now let~$q>1$. We build pairs of the indices in $X$. This gives us
$\lfloor|X|/2\rfloor$ pairs of indices and an additional single index,
if $|X|$ is odd. Due to Lemma~\ref{lem:semilinear-set-intersection} we
get for each such pair $(x,y)$ of indices an index set $H_{x,y}$ with
\begin{equation}\label{ineqn-h-xy}
|H_{x,y}|\le p^2(\ell+1)^{2k}
\end{equation}
and for each $h\in H_{x,y}$ sets
$C_h,\,P_h\subseteq\nat^k$ with
\begin{equation}\label{ineqns-ph-ch}
||P_h||\le a_kn^{3k+1},\qquad
||C_h||\le (a_kn^{3k+1}+1)\ell,
\end{equation}
and $$\left(\bigcup_{i\in I_x}L(C_i,P_i)\right)\cap\left(\bigcup_{j\in
    I_y}L(C_j,P_j)\right)=\bigcup_{h\in H_{x,y}}L(C_h,\,P_h).$$ So we
have such a semilinear set for each of our pairs of indices and
additionally a semilinear set for a single index out of $X$, if $|X|$
is odd. If we now intersect these $\lceil|X|/2\rceil$ semilinear sets,
we get~$\bigcap_{x\in X}\left(\bigcup_{i\in
    I_x}L(C_i,P_i)\right)$. Because of
$\lceil\log_2\lceil|X|/2\rceil\rceil=\lceil\log_2|X|\rceil-1=q-1$, we
can build this intersection by induction. This gives us an index set $J$ and for each
$j\in J$ sets $C_j,\,P_j\subseteq\nat^k$ with
$$\bigcap_{x\in X}\left(\bigcup_{i\in I_x}L(C_i,P_i)\right)=\bigcup_{j\in J}L(C_j,P_j).$$
To get a bound for $|J|$ we use Inequality~\ref{eqn:est-J}, where we replace $q$ by $q-1$. Inequalities~\ref{ineqn-h-xy} and~\ref{ineqns-ph-ch} give us bounds for $p$, $\ell$, and~$n$. So we have
\begin{align*}
|J|&\le(p^2(\ell+1)^{2k})^{2^{q-1}}((a_kn^{3k+1}+1)\ell+1)^{k\cdot 2^{q}}(a_k^2n^{3k+1}+1)^{4(3k+2)^{q}}\\
&\le p^{2^q}(\ell+1)^{k\cdot2^q}((a_kn^{3k+1}+1)(\ell+1))^{k\cdot 2^{q}}(a_kn+1)^{4(3k+1)(3k+2)^{q}}.
\end{align*}
By ordering the factors we get the upper bound
$$p^{2^q}(\ell+1)^{k(2^q+2^q)}(a_kn+1)^{(3k+1)k\cdot 2^q+4(3k+1)(3k+2)^{q}}
=p^{2^q}(\ell+1)^{k\cdot 2^{q+1}}(a_kn+1)^{4((3k+1)k\cdot 2^{q-2}+(3k+1)(3k+2)^{q})}.$$
Then
$$(3k+1)k\cdot 2^{q-2}+(3k+1)(3k+2)^{q}\le(3k+2)^{q}+(3k+1)(3k+2)^{q}=(3k+2)^{q+1}$$
gives us
$$|J|\le p^{2^q}(\ell+1)^{k\cdot 2^{q+1}}(a_kn+1)^{4(3k+2)^{q+1}}.$$
For each $j\in J$ we get
$$||P_j||\le(a_k^2n^{3k+1})^{(3k+1)^{q-1}}\le(a_kn)^{(3k+1)^q}$$
and
$$||C_j||\le(a_k^2n^{3k+1}+1)^{(3k+2)^{q-1}}(a_kn^{3k+1}+1)\ell
\le(a_kn+1)^{(3k+1)(3k+2)^{q-1}+3k+1}\ell.$$
Because of
$(3k+1)(3k+2)^{q-1}+3k+1\le(3k+1)(3k+2)^{q-1}+(3k+2)^{q-1}=(3k+2)^q$
we finally obtain the bound 
$||C_j||\le(a_kn+1)^{(3k+2)^q}\ell$.
\end{proof}

\subsection{Complementation of semilinear sets}

The next Boolean operation is the complementation.
Our result is based on~\cite[Lemma~ 6.6, Lemma~6.8, and Lemma~6.9]{GiSp64},
which we slightly adapt. First we complement a linear set where the
constant is the null-vector and the periods are linearly independent in Lemma~\ref{lem:compl-lin}. We continue
by complementing a linear set with an arbitrary constant and linearly independent periods in
Corollary~\ref{lem:compl-lin1}. To complement a semilinear set where all the period sets are linearly independent in
Theorem~\ref{thm:semilinear-complementation} we use DeMorgan's law: a
semilinear set is a finite union of linear sets, so the complement is
the intersection of the complements of the linear sets. For this
intersection we use
Theorem~\ref{theo:semilinear-set-intersection-many}. Then we convert an arbitrary linear set to a semilinear set with
linearly independent period sets in Lemma~\ref{lem:linear-independent}. Finally we insert the bounds from Lemma~\ref{lem:linear-independent}
into the bounds from Theorem~\ref{thm:semilinear-complementation} to complement an arbitrary semilinear set in Theorem~\ref{thm:semilinear-complementation1}.

\begin{lemma}\label{lem:compl-lin}
  Let $n,\,k\geq 1$, and $P\subseteq\nat^k$ be linearly independent with
  $||P||\leq n$. Then there exists an index set~$I$ with $|I|\le
  2^k+k-1$ such that, for each $i\in I$, there are
  subsets $P_i,\,C_i\subseteq\nat^k$ with
  $${||P_i||,\,||C_i||\le(2k+1)k^{k/2}n^k}$$ and
  $\nat^k\setminus L(\vec{0},P)=\bigcup_{i\in I}L(C_i,P_i)$.
\end{lemma}

\begin{proof}
  Let $P=\{\vec{x}_1,\vec{x}_2,\dots,\vec{x}_p\}$. 
  Since the vectors in $P$ are linearly independent, we conclude
  $p\leq k$. For~$i\in\{1,2,\dots,k\}$ let ${\vec{e}_i\in\nat^k}$ be the unit vector
  defined by ${(\vec{e}_i)_i=1}$ and ${(\vec{e}_i)_j=0}$ for $i\neq
  j$. By elementary vector space theory there exist
  ${\vec{x}_{p+1},\vec{x}_{p+2},\dots,\vec{x}_k\in\{\vec{e}_1,\vec{e}_2,\dots,\vec{e}_k\}}$
  such that ${\vec{x}_1,\vec{x}_2,\dots,\vec{x}_k}$ are linearly
  independent. Let~$\Delta$ be the absolute value of the determinant
  of the matrix ${(\vec{x}_1\mid \vec{x}_2\mid\dots\mid\vec{x}_k)}$.
    Moreover, let 
    $\pi:\nat^k\times\nat^k\to\nat^k$ be the projection on the first
    factor. For ${J,\,K\subseteq\{1,2,\dots,k\}}$ 
     we define
$$
A_J=\{\,(\vec{y},\vec{a})\in\nat^k\times\nat^k\mid a_j>0\textup{ for all
    }j\in J\,\}
$$
and
    \[B_K=\sset{(\vec{y},\vec{a})\in\nat^k\times\nat^k}{\Delta \vec{y}+\sum_{i\in K}a_i\vec{x}_i=\sum_{i\in\{1,\dots,k\}\setminus K}a_i\vec{x}_i}.\]
    Let~$Q_K$ and~$D_{K,J}$ be the sets of minimal elements of
    $B_K\setminus\{\vec{0}\}$ and $B_K\cap A_J$. Looking at the proof
    of~\cite[Theorem~6.1]{GiSp64} we see that $B_K\cap
    A_J=L(D_{K,J},Q_K)$. The linearity of $\pi$ implies 
$$\pi(B_K\cap A_J)=L(\pi(D_{K,J}),\pi(Q_K)).$$ Next define
$$B'_K=\left\{\,(\vec{y},\vec{a})\in\nat^k\times\nat^k\mathrel{\left\vert\vphantom{\sum_{i\in K}}\right.} \vec{y}+\sum_{i\in K}a_i\vec{x}_i=\sum_{i\in\{1,\dots,k\}\setminus K}a_i\vec{x}_i\right.
\left.\textup{ and } \vec{y}=\Delta \vec{z}\textup{ for some }\vec{z}\in\nat^k\,\vphantom{\sum_{i\in K}}\right\}$$
    and~$Q'_K$
    and $D'_{K,J}$ to be the sets of minimal elements of
    $B'_K\setminus\{0\}$ and $B'_K\cap A_J$. Then the mapping ${f:B'_K\to
    B_K,}$ defined \textit{via} $(\vec{y},\vec{a})\mapsto(\vec{y}/\Delta,\vec{a})$ is a bijection. The proof
    of~\cite{GaSi78} and Theorem~\ref{thm:hadamard} show that
\[||Q'_K||,\,||D'_{K,J}||\le(2k\Delta+1)k^{k/2}n^k\le\Delta\cdot(2k+1)k^{k/2}n^k.\]
With $Q_K=f(Q'_K)$ and $D_{K,J}=f(D'_{K,J})$ we get 
\[||\pi(Q_K)||,||\pi(D_{K,J})||\le(2k+1)k^{k/2}n^k.\]
Set $G_1=\bigcup_{\emptyset\ne K\subseteq\{1,\dots,k\}}\pi(B_K\cap A_K)$.

Because ${\vec{x}_1,\vec{x}_2,\dots,\vec{x}_k}$ are linearly
independent every $\vec{y}\in\nat^k$ can be written uniquely as
$\vec{y}=\sum_{i=1}^k\lambda_{y,i}\vec{x}_i$ with
$\lambda_{y,i}\in\mathbb{Q}$, for $i\in\{1,2,\dots,k\}$. Then
\mbox{$\vec{y}\in L(\vec{0},P)$} if and only if $\lambda_{y,i}\in\nat$, for
every $i\in\{1,2,\dots,p\}$ and~$\lambda_{y,i}=0$, for every~$i\in\{p+1,p+2,\dots,k\}$. In the proof of \cite[Lemma~6.7]{GiSp64} it
was shown that~$\Delta \vec{y}$ can be written uniquely as $\Delta
\vec{y}=\sum_{i=1}^k\mu_{y,i}\vec{x}_i$ with $\mu_{y,i}\in\mathbb{Z}$,
for $i\in\{1,2,\dots,k\}$. Because of $\lambda_{y,i}=\mu_{y,i}/\Delta$
we get that $\vec{y}\in L(\vec{0},P)$ if and only if~$\mu_{y,i}$ is a
non-negative multiple of~$\Delta$, for every $i\in\{1,2,\dots,p\}$ and
$\mu_{y,i}=0$, for every $i\in\{p+1,p+2,\dots,k\}$. The set~$G_1$
consists of all $\vec{y}\in\nat^k$ such that at least one of
the~$\mu_{y,i}$ is negative. This implies $G_1\subseteq\nat^k\setminus
L(\vec{0},P)$.

Now we set $G_2=\bigcup_{i=p+1}^k\pi(B_\emptyset\cap A_{\{i\}})$. This
set consists of all $\vec{y}\in\nat^k$ such that all the~$\mu_{y,i}$
are non-negative and there exists $i\in\{p+1,p+2,\dots,k\}$ such
that~$\mu_{y,i}$ is positive. This implies
$G_2\subseteq\nat^k\setminus L(\vec{0},P)$.

For $i\in\{1,2,\dots,p\}$ and $r\in\{0,1,\dots,\Delta-1\}$ we set
$$E_{i,r}=\sset{(\vec{y},\vec{a})\in\nat^k\times\nat^p}{\Delta
  \vec{y}=\sum_{j=1}^pa_j\vec{x}_j\textup{ and }a_i\bmod \Delta=r}.$$ Let~$R_{i,r}$
be the set of minimal elements of~$E_{i,r}\setminus\{\vec{0}\}$. According to the proof
of~\cite[Theorem~6.1]{GiSp64} we get~$E_{i,r}=L(R_{i,r},R_{i,0})$, for
$r>0$. We set $\pi':\nat^k\times\nat^p\to\nat^k$ to be the projection
on the first factor. Then~$\pi'(E_{i,r})=L(\pi'(R_{i,r}),\pi'(R_{i,0}))$, for $r>0$, and
$$\bigcup_{r=1}^{\Delta-1}\pi'(E_{i,r})=L(\bigcup_{r=1}^{\Delta-1}\pi'(R_{i,r}),\pi'(R_{i,0})).$$ 
Let
$(\vec{y},\vec{a})\in R_{i,r}$. Then we have $||\vec{a}||\le\Delta$. This implies $||\vec{y}||\le pn$. So we obtain $||\pi'(R_{i,r})||\le pn$. Define~$G_3=\bigcup_{i=1}^p\bigcup_{r=1}^{\Delta-1}\pi'(E_{i,r})$. This is
the set of all vectors $\vec{y}\in\nat^k$ such that $\mu_{y,j}=0$, for every
$j\in\{p+1,p+2,\dots,k\}$, $\mu_{y,j}\ge 0$, for every $j\in\{1,2,\dots,p\}$,
and~$\mu_{y,j}$ is not divisible by~$\Delta$ for at least one
$j\in\{1,2,\dots,p\}$. Thus we have $G_1\cup G_2\cup G_3=\nat^k\setminus
L(\vec{0},P)$.
\end{proof}

The next lemma gives a size estimation for the set $\nat^k\setminus
L(\vec{x}_0,P)$, for an arbitrary vector~$\vec{x}_0$, instead of the
null-vector, as in the previous theorem.

\begin{lemma}\label{lem:compl-lin1}
  Let $k\geq 1$, subset $P\subseteq\nat^k$ be linearly independent,
  and $\vec{x}_0\in\nat^k$. Then there exists an index set~$I$ with
  $|I|\le2^k+2k-1$ such that, for each $i\in I$, there are
  $P_i,\,C_i\subseteq\nat^k$ with
\begin{align*}
||P_i||&\le(2k+1)k^{k/2}(||P||+1)^k,\\
||C_i||&\le(2k+1)k^{k/2}(||P||+1)^k+||\vec{x}_0||,\\
|C_i|&\le\max(4^kk^{k^2/2+k}(||P||+1)^{k^2},||\vec{x}_0||),
\end{align*}
and $\nat^k\setminus L(\vec{x}_0,P)=\bigcup_{i\in I}L(C_i,P_i)$.
\end{lemma}

\begin{proof}
For $j\in\{1,2,\dots,k\}$ let
$$D_j=\{\,\vec{y}\in\nat^k\mid y_\ell=0\textup{ for }\ell\ne j\textup{
  and }y_j<(\vec{x}_0)_j\,\}$$ and
$Q_j=\{\vec{e}_1,\dots,\vec{e}_{j-1},\vec{e}_{j+1},\dots,\vec{e}_k\}$,
where the $\vec{e}_\ell$ are defined as in the proof of
Lemma~\ref{lem:compl-lin}. Define the set~$G=\bigcup_{j=1}^kL(D_j,Q_j)$. This is the set of all
$\vec{y}\in\nat^k$ such that $\vec{x}_0\le \vec{y}$ is false. So we
have $G\subseteq\nat^k\setminus L(\vec{x}_0,P)$.

Now let $Y=\{\,\vec{y}\in\nat^k\mid \vec{x}_0\le\vec{y}\,\}$. Then
$\nat^k\setminus L(\vec{x}_0,P)=G\cup(Y\setminus L(\vec{x}_0,P))$. We
have $Y\setminus L(\vec{x}_0,P)=(\nat^k\setminus
L(\vec{0},P))+\vec{x}_0$. Due to Lemma~\ref{lem:compl-lin} we have
an index set~$J$ with $|J|\le 2^k+k-1$ and for each $j\in J$ subsets
$C_j,\,P_j\subseteq\nat^k$ with
${||C_j||,\,||P_j||\le(2k+1)k^{k/2}(||P||+1)^k}$ such that $\nat^k\setminus
L(\vec{0},P)=\bigcup_{j\in J}L(C_j,P_j)$. This gives us
$(\nat^k\setminus L(\vec{0},P))+\vec{x}_0=\bigcup_{j\in
  J}L(C_j+\vec{x}_0,P_j)$. Because of $C_j\subseteq\nat^k$ we obtain
	$$|C_j+\vec{x}_0|=|C_j|\le((2k+1)k^{k/2}(||P||+1)^k+1)^k
	\le(4k^{k/2+1}(||P||+1)^k)^k=4^kk^{k^2/2+k}(||P||+1)^{k^2}.$$
This proves the stated claim.
\end{proof}

Now we are ready to deal with the complement of a semilinear set with
linearly independent period sets.

\begin{theorem}\label{thm:semilinear-complementation}
  Let $k\geq 1$ and $\bigcup_{i\in I}L(\vec{x}_i,P_i)$ be a semilinear
  subset of $\nat^k$ with $I\ne\emptyset$ and linearly independent
  sets $P_i$. We set $n=\max_{i\in I}||P_i||$ and $\ell=\max_{i\in
    I}||\vec{x}_i||$. Define $q=\lceil\log_2|I|\rceil$.  Then there
  exists an index set~$J$
  with $$|J|\le(4k(n+1))^{5(k+2)(3k+2)^{q+1}}(\ell+1)^{k\cdot
    2^{q+1}}$$ such that, for each $j\in J$, there are
	$P_j,\,C_j\subseteq\nat^k$ with
\begin{align*}
||P_j||&\le(4k(n+1))^{(k+2)(3k+1)^q},\\
||C_j||&\le(4k(n+1))^{(k+2)(3k+2)^q+k}(\ell+1),
\end{align*}
and $\nat^k\setminus\bigcup_{i\in I}L(\vec{x}_i,P_i)=\bigcup_{j\in J}L(C_j,P_j)$.
\end{theorem}

\begin{proof}
  Due to DeMorgan's law we have
  $$\nat^k\setminus\bigcup_{i\in
    I}L(\vec{x}_i,P_i)=\bigcap_{i\in I}\left(\nat^k\setminus
  L(\vec{x}_i,P_i)\right).
  $$
  Because of
  Lemma~\ref{lem:compl-lin1} for every $i\in I$ there exists an
  index set $H_i$ with ${|H_i|\le 2^{k+1}}$ such that, for each~$h\in
  H_i$, there are $C_h,\,P_h\subseteq\nat^k$ with
\begin{align*}
||P_h||&\le(2k+1)k^{k/2}(n+1)^k\le3k^{k/2+1}(n+1)^k,\\
||C_h||&\le(2k+1)k^{k/2}(n+1)^k+\ell\le3k^{k/2+1}(n+1)^k+\ell,
\end{align*}
and $\nat^k\setminus L(\vec{x}_i,P_i)=\bigcup_{h\in H_i}L(C_h,P_h)$.
Theorem~\ref{theo:semilinear-set-intersection-many} gives us an index set $J$ and for each $j\in J$ sets $C_j,\,P_j\subseteq\nat^k$ with
$$\bigcup_{j\in J}L(C_j,P_j)=\bigcap_{i\in I}\left(\bigcup_{h\in H_i}L(C_h,P_h)\right)=\nat^k\setminus\bigcup_{i\in I}L(\vec{x}_i,P_i)$$
and
\begin{align*}
  |J|&\le(2^{k+1})^{2^q}(3k^{k/2+1}(n+1)^k+\ell+1)^{k\cdot 2^{q+1}}(4^{k+1}k^{k/2}\cdot3k^{k/2+1}(n+1)^k+1)^{4(3k+2)^{q+1}}\\
  &\le2^{(k+1)\cdot2^q}(4k^{k/2+1}(n+1)^k(\ell+1))^{k\cdot 2^{q+1}}(4^{k+2}k^{k+1}(n+1)^k)^{4(3k+2)^{q+1}}.
\end{align*}
Now we order the factors and get that this is less than or equal to
$$2^{(k+1)\cdot 2^q+2k\cdot 2^{q+1}+8(k+2)(3k+2)^{q+1}}(k(n+1))^{(k+1)k\cdot 2^{q+1}+4(k+1)(3k+2)^{q+1}}(\ell+1)^{k\cdot 2^{q+1}}.$$
Because of $(5k+1)\le(k+2)(3k+2)$ we have
$$(k+1)\cdot 2^q+2k\cdot 2^{q+1}+8(k+2)(3k+2)^{q+1}
=(5k+1)\cdot 2^q+8(k+2)(3k+2)^{q+1}\le 9(k+2)(3k+2)^{q+1}.$$
Furthermore $k\cdot 2^{q+1}\le(3k+2)^{q+1}$ gives us $$(k+1)k\cdot 2^{q+1}+4(k+1)(3k+2)^{q+1}\le 5(k+1)(3k+2)^{q+1}.$$ So we get
$$|J|\le2^{9(k+2)(3k+2)^{q+1}}(k(n+1))^{5(k+1)(3k+2)^{q+1}}(\ell+1)^{k\cdot 2^{q+1}}
\le(4k(n+1))^{5(k+2)(3k+2)^{q+1}}(\ell+1)^{k\cdot 2^{q+1}}.$$
For each $j\in J$ we have
$$||P_j||\le(4^{k+1}k^{k/2}\cdot3k^{k/2+1}(n+1)^k)^{(3k+1)^q}\le(4k(n+1))^{(k+2)(3k+1)^q}$$ and
\begin{align*}
||C_j||&\le(4^{k+1}k^{k/2}\cdot3k^{k/2+1}(n+1)^k+1)^{(3k+2)^q}(3k^{k/2+1}(n+1)^k+\ell)\\
&\le(4^{k+2}k^{k+1}(n+1)^k)^{(3k+2)^q}(4k^{k/2+1}(n+1)^k(\ell+1)).
\end{align*}
From 
$k^{(k+1)(3k+2)^q}k^{k/2+1}=k^{(k+1)(3k+2)^q+k/2+1}
\le k^{(k+1)(3k+2)^q+k+(3k+2)^q}=k^{(k+2)(3k+2)^q+k}$ 
we finally deduce
$||C_j||\le(4k(n+1))^{(k+2)(3k+2)^q+k}(\ell+1)$.
\end{proof}

Next we convert an arbitrary linear set to a semilinear set with
linearly independent period sets. The idea is the following: If the
periods are linearly dependent we can rewrite our linear set as a
semilinear set, where in each period set one of the original periods
is removed. By doing this inductively the period sets get smaller and
smaller until they are finally linearly independent.

\begin{lemma}\label{lem:linear-independent}
  Let $L(\vec{x}_0,P)$ be a linear subset of $\nat^k$ for some $k\geq
  1$. We set $m=|P|$ and $n=||P||$. Then there exists an index set~$I$
  with
  $$|I|\le(m+1)!\cdot m!/2^m\cdot(k^{k/2}n^k+1)^{m-1}$$
  and, for each $i\in I$, a linearly independent subset
  $P_i\subseteq\nat^k$ with $||P_i||\le n$ and a vector
  $\vec{x}_i\in\nat^k$ with
  $$||\vec{x}_i||\le||\vec{x}_0||+(m+1)(m+2)/2\cdot k^{k/2}n^{k+1}$$
  such that $\bigcup_{i\in I}L(\vec{x}_i,P_i)=L(\vec{x}_0,P)$.
\end{lemma}

\begin{proof}
  We prove this by induction on $m$. The statement of the lemma is
  clearly true for~$m=0$ or~$m=1$. So let $m\ge 2$ now. If $P$ is
  linearly independent the statement of the lemma is trivial. Thus we
  assume $P$ to be linearly dependent. Then there exists
  $p\in\{1,2,\dots,\lfloor m/2\rfloor\}$ and pairwise different
  vectors $x_1,x_2,\dots,x_p,y_1,y_2,\dots,y_{m-p}\in P$ such that
  $X=\sset{\vec{a}\in\nat^m\setminus\{\vec{0}\}}{H\cdot\vec{a}=\vec{0}}$,
  where~$H\in\ints^{k\times m}$ is the matrix
  $(x_1|x_2|\dots|x_p|-y_1|-y_2|\dots|-y_{m-p})$, is not empty. Let
  $\vec{a}$ be a minimal element of $X$. From
  Theorem~\ref{thm:size-of-minimal-solutions} and
  Theorem~\ref{thm:hadamard} we deduce 
  $||\vec{a}||\le(m+1)k^{k/2}n^k$. For $j\in\{1,2,\dots,p\}$ let
  $C_j=\sset{\vec{x}_0+\lambda\vec{x}_j}{\lambda\in\{0,1,\dots,a_j-1\}}$,
  if $a_j>0$, and $C_j=\{\vec{x}_0\}$, otherwise. Furthermore let
  $Q_j=P\setminus\{\vec{x}_j\}$. In the proof
  of~\cite[Lemma~6.6]{GiSp64} it was shown that
  $\bigcup_{j=1}^pL(C_j,Q_j)=L(\vec{x}_0,P)$. We can rewrite this set
  as $\bigcup_{j\in\{1,2,\dots,p\},\,\vec{c}\in
    C_j}L(\vec{c},Q_j)$. Here the size of the index set is smaller than or
  equal to $m/2\cdot||\vec{a}||\le(m+1)m/2\cdot k^{k/2}n^k$ and for each
  such $\vec{c}$ we have $||\vec{c}||\le||\vec{x}_0||+||\vec{a}||\cdot
  n\le||\vec{x}_0||+(m+1)k^{k/2}n^{k+1}$. Because of $|Q_j|=m-1$ for
  each $j\in\{1,2,\dots,p\}$ and~$\vec{c}\in C_j$, by induction, there
  exists an index set $I_{j,\vec{c}}$ with
$$|I_{j,\vec{c}}|\le m!\cdot(m-1)!/2^{(m-1)}\cdot(k^{k/2}n^k+1)^{m-2}$$
	and, for each $i\in I_{j,\vec{c}}$, a linearly independent subset
	$R_i\subseteq\nat^k$ with $||R_i||\le n$ and a vector $\vec{z}_i\in\nat^k$ with
$$||\vec{z}_i||\le||\vec{c}||+m(m+1)/2\cdot k^{k/2}n^{k+1}$$
such that $\bigcup_{i\in I_{j,\vec{c}}}L(\vec{z}_i,R_i)=L(\vec{c},Q_j)$. This gives us
$$\bigcup_{j\in\{1,2,\dots,p\},\,\vec{c}\in C_j,\,i\in I_{j,\vec{c}}}L(\vec{z}_i,R_i)=L(\vec{x}_0,P).$$
The size of this index set is smaller than or equal to
$$(m+1)m/2\cdot k^{k/2}n^k\cdot m!\cdot(m-1)!/2^{(m-1)}\cdot(k^{k/2}n^k+1)^{m-2}
\le(m+1)!\cdot m!/2^m\cdot(k^{k/2}n^k+1)^{m-1}.$$
With
$$||\vec{z}_i||\le||\vec{x}_0||+(m+1)k^{k/2}n^{k+1}+m(m+1)/2\cdot k^{k/2}n^{k+1}
\le||\vec{x}_0||+(m+1)(m+2)/2\cdot k^{k/2}n^{k+1}$$
the lemma is proved.
\end{proof}

With Theorem~\ref{thm:semilinear-complementation} and Lemma~\ref{lem:linear-independent} we are ready to complement an arbitrary semilinear set.

\begin{theorem}\label{thm:semilinear-complementation1}
  Let $k\geq 1$ and $\bigcup_{i\in I}L(\vec{x}_i,P_i)$ be a semilinear
  subset of $\nat^k$ with $I\ne\emptyset$. We set $n=\max_{i\in I}||P_i||$, $m=\max_{i\in I}|P_i|$, and $\ell=\max_{i\in
    I}||\vec{x}_i||$. Define $b(k,n,m,I)$ as
		$$\left(\sqrt{k}(n+2)\right)^{k\cdot\log_2(3k+2)\cdot(3m+1)+3}\cdot(3k+2)^{-(2\log_2(e)+1)m+7}\cdot|I|^{\log_2(3k+2)}.$$		
		Then there
  exists an index set~$J$ with
	$$|J|\le2^{b(k,n,m,I)}\cdot(\ell+2)^{\left(\sqrt{k}(n+2)\right)^{k\cdot(3m+1)+8}\cdot\left(2e^2\right)^{-m}\cdot|I|}$$
	such that, for each $j\in J$, there are
	$P_j,\,C_j\subseteq\nat^k$ with
\begin{align*}
||P_j||&\le2^{b(k,n,m,I)},\\
||C_j||&\le2^{b(k,n,m,I)}\cdot(\ell+1),
\end{align*}
and $\nat^k\setminus\bigcup_{i\in I}L(\vec{x}_i,P_i)=\bigcup_{j\in J}L(C_j,P_j)$.
\end{theorem}

\begin{proof}
Because of Lemma~\ref{lem:linear-independent} there exists an index set~$H\ne\emptyset$ with
	$$|H|\le(m+1)!\cdot m!/2^m\cdot(k^{k/2}n^k+1)^{m-1}\cdot|I|$$
	and, for each $h\in H$, a linearly independent subset
	$P_h\subseteq\nat^k$ with $||P_h||\le n$ and a vector $\vec{x}_h\in\nat^k$ with
$$||\vec{x}_h||\le\ell+(m+1)(m+2)/2\cdot k^{k/2}n^{k+1}$$
such that $\bigcup_{h\in H}L(\vec{x}_h,P_h)=\bigcup_{i\in I}L(\vec{x}_i,P_i)$.
With Stirling's formula we get 
%\begin{align*}
%{(m+1)!} & \le(m+1)^{m+3/2}e^{-m}\\
%\noalign{\hbox{and}}
%m! & \le(m+1)^{m+1/2}e^{-m+1}.
%\end{align*}
$${(m+1)!}  \le(m+1)^{m+3/2}e^{-m}
\qquad\textup{and}\qquad
m!  \le(m+1)^{m+1/2}e^{-m+1}.$$
This gives us
$(m+1)!\cdot m!\le(m+1)^{2m+2}e^{-2m+1}\le((n+1)^k+1)^{2m+2}e^{-2m+1}$
and we get
\begin{align*}
|H|&\le(m+1)!\cdot m!/2^m\cdot(k^{k/2}n^k+1)^{m-1}\cdot|I|\\
&\le((n+1)^k+1)^{2m+2}e^{-2m+1}/2^m\cdot(k^{k/2}n^k+1)^{m-1}\cdot|I|\\
&\le(k^{k/2}(n+2)^k)^{2m+2}e^{-2m+1}/2^m\cdot(k^{k/2}(n+2)^k)^{m-1}\cdot|I|\\
&=e\cdot\left(\sqrt{k}(n+2)\right)^{k\cdot(3m+1)}\cdot\left(2e^2\right)^{-m}\cdot|I|.
\end{align*}
We shall use Theorem~\ref{thm:semilinear-complementation} to get
upper bounds for the complement. So we set
$q=\lceil\log_2|H|\rceil$. In all three bounds of
Theorem~\ref{thm:semilinear-complementation} the exponent of $4k(n+1)$
is bounded from above by $(3k+2)^{q+3}$.  We have
\begin{align*}
(3k+2)^{q+3}
& \le(3k+2)^{\log_2|H|+4}\\
& =(3k+2)^4\cdot|H|^{\log_2(3k+2)}\\
& \le(3k+2)^4\cdot\left(e\cdot\left(\sqrt{k}(n+2)\right)^{k\cdot(3m+1)}\cdot\left(2e^2\right)^{-m}\cdot|I|\right)^{\log_2(3k+2)}\\
&=\left(\sqrt{k}(n+2)\right)^{k\cdot\log_2(3k+2)\cdot(3m+1)}\cdot(3k+2)^{-(2\log_2(e)+1)m+\log_2(e)+4}\cdot|I|^{\log_2(3k+2)}\\
& \le\left(\sqrt{k}(n+2)\right)^{-3}\cdot(3k+2)^{-1}\cdot b(k,n,m,I).
\end{align*}
Because of $\log_2(4k(n+1))\le\left(\sqrt{k}(n+2)\right)^{2}$ we get
\begin{equation}\label{bound}
(4k(n+1))^{(3k+2)^{q+3}}\le2^{\left(\sqrt{k}(n+2)\right)^{-1}\cdot(3k+2)^{-1}\cdot b(k,n,m,I)}.
\end{equation}
For the sets $P_j$ from Theorem~\ref{thm:semilinear-complementation} this implies $||P_j||\le2^{b(k,n,m,I)}$.
For each~$h\in H$ we have
\begin{align*}
||\vec{x}_h||+1&\le\ell+(m+1)(m+2)/2\cdot k^{k/2}n^{k+1}+1\\
&\le\ell+k^{k/2}(n+2)^k(n+3)^kn^{k+1}/2+1\\
&\le\ell+\left(\sqrt{k}n(n+2)(n+3)\right)^{k+1}/2+1\\
&\le\ell+\left(\sqrt{k}(n+2)^3\right)^{k+1}\\
&=\ell+2^{(k+1)\cdot\log_2\left(\sqrt{k}(n+2)^3\right)}\\
&\le2^{(k+1)\cdot\log_2\left(\sqrt{k}(n+2)^3\right)}\cdot(\ell+1).
\end{align*}
From
$(3k+2)^{q+3}\le\left(\sqrt{k}(n+2)\right)^{-3}\cdot(3k+2)^{-1}\cdot b(k,n,m,I)$
we get
$$(k+1)\cdot\log_2\left(\sqrt{k}(n+2)^3\right)\le\left(\sqrt{k}(n+2)\right)^{3}\le(3k+2)^{-1}\cdot b(k,n,m,I).$$
This leads to $||\vec{x}_h||+1\le2^{(3k+2)^{-1}\cdot
  b(k,n,m,I)}\cdot(\ell+1)$. Together with Inequality~(\ref{bound})
this implies 
$$||C_j||\le2^{b(k,n,m,I)}\cdot(\ell+1)$$ for the sets
$C_j$ from Theorem~\ref{thm:semilinear-complementation}, because of
$3k+2\ge 2$.  For each $h\in H$ we have
\begin{align*}
\left(||\vec{x}_h||+1\right)^{k\cdot2^{q+1}}&\le\left(\ell+2^{(k+1)\cdot\log_2\left(\sqrt{k}(n+2)^3\right)}\right)^{4k\cdot|H|}\\
&\le(\ell+2)^{(k+1)\cdot\log_2\left(\sqrt{k}(n+2)^3\right)\cdot 4k\cdot e\cdot\left(\sqrt{k}(n+2)\right)^{k\cdot(3m+1)}\cdot\left(2e^2\right)^{-m}\cdot|I|}\\
&\le(\ell+2)^{\left(\sqrt{k}(n+2)\right)^{k\cdot(3m+1)+8}\cdot\left(2e^2\right)^{-m}\cdot|I|}
\end{align*}
because
$$4e\cdot k(k+1)\cdot\log_2\left(\sqrt{k}(n+2)^3\right)\le12\cdot2k^2\cdot3\cdot\log_2\left(\sqrt{k}(n+2)\right)
\le72k^2\cdot\sqrt{k}\cdot(n+2)
\le2^7\cdot\left(\sqrt{k}\right)^5 (n+2).$$
With Inequality~(\ref{bound}) we get
$|J|\le2^{b(k,n,m,I)}\cdot(\ell+2)^{\left(\sqrt{k}(n+2)\right)^{k\cdot(3m+1)+8}\cdot\left(2e^2\right)^{-m}\cdot|I|}$
for the set $J$ from Theorem~\ref{thm:semilinear-complementation}. This proves our theorem.
\end{proof}

The size of the resulting index set and the norms for the resulting
periods and constants are bounded from above by
$2^{(\nu+2)^{O(m)}\cdot|I|^{\log(3k+2)}}$, if $k$ is constant and, as
before, $\nu$ is the maximum of $n$ and $\ell$. So we observe that the
size increase is exponential in $\nu$ and $|I|$ and double exponential
in $m$.

\subsection{Inverse homomorphism on semilinear sets}

Finally, we consider the descriptional complexity of the inverse
homomorphism. We follow the lines of the proof on
the inverse homomorphism closure given in~\cite{GiSp64}. Since inverse
homomorphism commutes with union, we only need to look at linear sets.
The vectors in the pre-image of a linear set, with respect to a
homomorphism, can be described by a system of linear equations. Now we
use the same techniques as in the proof of
Theorem~\ref{thm:semilinear-set-intersection}: out of the minimal
solutions of the system of equations we can build periods and constants
of a semilinear description of the pre-image. With
Theorems~\ref{thm:size-of-minimal-solutions} and~\ref{thm:hadamard} we
estimate the size of the minimal solutions to get upper bounds for the
norms of the resulting periods and constants.

\begin{theorem}\label{thm:seminlinear-set-invhom}
  Let $k_1,\,k_2\geq 1$ and $\bigcup_{i\in I}L(\vec{c}_i,P_i)$ be a
  semilinear subset of $\nat^{k_2}$. We set $n=\max_{i\in I}||P_i||$,\
  $m=\max_{i\in I}|P_i|$,\ and $\ell=\max_{i\in
    I}||\vec{c}_i||$. Moreover let $H\in\nat^{k_2\times k_1}$ be a
  matrix and $h:\nat^{k_1}\to\nat^{k_2}$ be the corresponding linear
  function $\vec{x}\mapsto H\vec{x}$. Then for every $i\in I$ there
  exist $Q_i,\,C_i\subseteq\nat^{k_1}$ with
\begin{align*}
||Q_i||&\le (k_1+m+1)k_2^{\min(k_1+m,k_2)/2}\cdot(||H||+1)^{\min(k_1,k_2)}(n+1)^{\min(m,k_2)},\\
||C_i||&\le (k_1+m+1)k_2^{\min(k_1+m,k_2)/2}\cdot(||H||+1)^{\min(k_1,k_2)}(n+1)^{\min(m,k_2)}\ell,
\end{align*}
and $h^{-1}\left(\bigcup_{i\in I}L(\vec{c}_i,P_i)\right)=\bigcup_{i\in
  I}L(C_i,Q_i)$.
\end{theorem}

\begin{proof}
  Let $i\in I$ be fixed and define $P_i$ to be $\{\vec{y}_1,\vec{y}_2,\dots,\vec{y}_p\}$. Then the
  set of vectors 
$$\sset{\vec{x}\in\nat^{k_1}}{H\vec{x}\in
  L(\vec{c}_i,P_i)}$$ is equal to
  $\sset{\vec{x}\in\nat^{k_1}}{\exists\lambda_1,\lambda_2,\dots,\lambda_p\in\nat:\,H\vec{x}=\vec{c}_i+\lambda_1\vec{y}_1+\lambda_2\vec{y}_2+\dots+\lambda_p\vec{y}_p}$.

  Now let $\tau:\nat^{k_1}\times\nat^p\to\nat^{k_1}$ be the projection
  on the first component and let \mbox{$J\in\ints^{k_2\times(k_1+p)}$} be the matrix~$J=(H\mid -\vec{y}_1\mid -\vec{y}_2\mid \cdots\mid -\vec{y}_p)$. We obtain
  \[\sset{\vec{x}\in\nat^{k_1}}{H\vec{x}\in
    L(\vec{c}_i,P_i)}=\tau\left(\sset{\vec{x}\in\nat^{k_1+p}}{J\vec{x}=\vec{c}_i}\right).\]
  Let $C\subseteq\nat^{k_1+p}$ be the set of minimal elements of
  $\sset{\vec{x}\in\nat^{k_1+p}}{J\vec{x}=\vec{c}_i}$ and
  $Q\subseteq\nat^{k_1+p}$ be the set of minimal elements of
  $\sset{\vec{x}\in\nat^{k_1+p}\setminus\{\vec{0}\}}{J\vec{x}=\vec{0}}$. In
  the proof of~\cite[Theorem~6.1]{GiSp64} it is shown that~$L(C,Q)=\sset{\vec{x}\in\nat^{k_1+p}}{J\vec{x}=\vec{c}_i}$. With
  $p\le m$ and
  $||P_i||\le n$, we derive from
  Theorems~\ref{thm:size-of-minimal-solutions} and~\ref{thm:hadamard}
  that
  \begin{align*}
	||Q||&\le (k_1+m+1)k_2^{\min(k_1+m,k_2)/2}\cdot(||H||+1)^{\min(k_1,k_2)}(n+1)^{\min(m,k_2)}
	\end{align*}
	With $||\vec{c}_i||\le\ell$ we get
	\[||C||\le
        (k_1+m+1)k_2^{\min(k_1+m,k_2)/2}\cdot(||H||+1)^{\min(k_1,k_2)}(n+1)^{\min(m,k_2)}\ell.\]
        Since~$\tau$ is linear we have~$L(\tau(C),\tau(Q))=
        \sset{\vec{x}\in\nat^{k_1}}{H\vec{x}\in L(\vec{c}_i,P_i)}$.
        Moreover, we have the inequalities~${||\tau(Q)||\le||Q||}$
        and~${||\tau(C)||\le||C||}$.
Because of
$h^{-1}\left(\bigcup_{i\in I}L(\vec{c}_i,P_i)\right)=\bigcup_{i\in I}h^{-1}\left(L(\vec{c}_i,P_i)\right)$
our theorem is proved.
\end{proof}

We see that the index set of the semilinear set is not changed under
inverse homomorphism. If~$k_1$ and~$k_2$ are constant, then the norms
of the periods and constants of the resulting semilinear set are in
$O\left((||H||+1)^{\min(k_1,k_2)}(m+1)(\nu+1)^{k_2+1}\right)$. Again~$\nu$
is the maximum of $n$ and $\ell$. Thus, the size increase for inverse
homomorphism is polynomial with respect to all parameters.

\bibliographystyle{eptcs}
\bibliography{doesntmatter}

\end{document}